\documentclass[10pt,a4paper]{article}
\usepackage{amsmath}
\usepackage{amsfonts}
\usepackage{amssymb}
\usepackage{amsthm,amscd}
\usepackage{enumerate}
\usepackage{xcolor}
\usepackage{comment}

\usepackage{titlesec}
\titleformat*{\section}{\bf\centering} % INDICATES THE FONT SIZE OF THE SECTIONS. 
\titleformat*{\subsection}{\centering\bf} % INDICATES THE FONT SIZE OF THE SUBSECTIONS. 
\titleformat*{\subsubsection}{\centering\it} % INDICATES THE FONT OF SUBSECTIONS.

\newtheorem{teo}{Theorem}[section]

\newtheorem{prop}[teo]{Proposition}
\newtheorem{cor}[teo]{Corollary}

\theoremstyle{remark}

\theoremstyle{definition}

\newcommand{\Hr}{\mathcal{C}}
\newcommand{\orb}{\mathcal{O}}
\newcommand{\ism}{\mathcal{I}}

\newcommand{\vs}{\vspace{.2cm}}
\newcommand{\cl}{{\rm cl}}
\newcommand{\be}{\begin{equation}}
\newcommand{\ee}{\end{equation}}
\begin{document}

\begin{center}

{\Large\bf A classification theorem for compact 
\vspace{.2cm}

Cauchy horizons in vacuum spacetimes}

\vspace{.5cm}

{\sc Ignacio Bustamante Bianchi}

{ibustamante@cmat.edu.uy}

\vspace{.3cm}

{\sc Mart\'in Reiris Ithurralde}

{mreiris@cmat.edu.uy}

\vspace{.4cm}

{\it Centro de Matem\'atica, Universidad de la Rep\'ublica,} 

{\it Montevideo, Uruguay.}
\vspace{.4cm}

\end{center}

\begin{abstract}
We establish a complete classification theorem for the topology and for the null generators of compact non-degenerate Cauchy horizons of time orientable smooth vacuum $3+1$-spacetimes.  We show that, either: (i) all generators are closed, or (ii) only two generators are closed and any other densely fills a two-torus, or (iii) every generator densely fills a two-torus, or (iv) every generator densely fills the horizon. We then show that, respectively to (i)-(iv), the horizon's manifold is either: (i') a Seifert manifold, or (ii') a lens space, or (iii') a two-torus bundle over a circle, or, (iv') a three-torus. All the four possibilities are known to arise in examples. In the last case, (iv), (iv'), we show in addition that the spacetime is indeed flat Kasner, thus settling a problem posed by Isenberg and Moncrief for ergodic horizons. The results of this article open the door for a full parameterization of the metrics of all vacuum spacetimes with a compact Cauchy horizon. The method of proof permits direct generalizations to higher dimensions. 
\end{abstract}

\section{Introduction}

The occurrence of Cauchy horizons in cosmological spacetimes beyond which predictability (from the initial data) fails, is one of the most intriguing features of the General Theory of Relativity closely related to the Strong Cosmic Censorship conjecture. In this regard, along the past decades there have been significant efforts to find necessary conditions under which they form and to provide a clear picture of their nature when they do. Motivated by these questions, in this paper we take advantage of recent results by the authors in \cite{reiris2020existence}, by Petersen in \cite{petersen2018wave} and by Petersen-R\'acz in \cite{petersen2018symmetries}, to establish a stringent list of the topologies and of the orbital type of the generators that non-degenerate Cauchy horizons can have. As a byproduct we prove that non-degenerate ergodic Cauchy horizons, that is, those having a dense generator, are just a quotient of the flat Kasner spacetime. This answer a question posed by Isenberg and Moncrief in \cite{moncrief2020symmetries}. We move now to describe in detail the setup and state the main results. We comment also on previous results in the literature.
\vs

We let $(M,g)$ be a smooth, connected time-oriented $3+1$-dimensional vacuum spacetime having a compact Cauchy horizon $\Hr$. We assume that the horizon $\Hr$ divides $M$ into two connected disjoint regions $I$ and $H$, i.e. $M\setminus \Hr=I\cup H$, where $H$ is a maximal globally hyperbolic spacetime with a compact and boundaryless Cauchy surface $\Sigma$. In this context $\Hr$ is a smooth \cite{larsson2015smoothness}, \cite{minguzzi2014completeness}, \cite{minguzzi2015area}, and totally geodesic null hypersurface, hence ruled by inextensible null geodesics called the {\it null generators}. We say that $\Hr$ is non-degenerate if there is at least one future or past incomplete generator (i.e. its affine length is finite). 

Under this setup, the following is the main result.

\begin{teo}\label{princ} Let $\Hr$ be a non-degenerate compact Cauchy horizon inside a smooth, time-orientable vacuum spacetime. Then, one of the following holds,
\begin{enumerate}[\rm (i)]
\item\label{i} all generators are closed,
\item\label{ii} only two generators are closed and every other generator densely fills a two-torus,
\item\label{iii} every generator densely fills a two-torus,
\item\label{iv} every generator densely fills the horizon.
\end{enumerate}
\end{teo}
A direct consequence will be the following topological classification of Cauchy horizons.
\begin{cor}\label{cprinc} Let $\Hr$ be a non-degenerate compact Cauchy horizon inside a smooth, time-orientable vacuum spacetime. Then, respectively to the cases {\rm (\ref{i})}-{\rm (\ref{iv})} in Theorem \ref{princ}, we have,
\begin{enumerate}[\rm (i')]
\item if {\rm (\ref{i})} holds, then $\Hr$ is a Seifert manifold,
\item if {\rm (\ref{ii})} holds, then $\Hr$ is a lens space,
\item if {\rm (\ref{iii})} holds, then $\Hr$ is a  $\mathbb{T}^2$-bundle over $\mathbb{S}^1$,
\item if {\rm (\ref{iv})} holds, then $\Hr$ is a three-torus $\mathbb{T}^3$. 
\end{enumerate}
\end{cor}    

All four possibilities in the theorem and the corollary are well known to arise in examples. For example, the Taub-NUT spacetime (see \cite{1967rta1.book..160M}) is an instance of (\ref{i}) and Gowdy Cauchy horizons are instances of (\ref{ii}) and (\ref{iii}) (see \cite{Chrusciel:2003jj} and references therein). An elaborated discussion of possibilities can be found in the comprehensive article of Chrusciel and Rendall \cite{1995AnPhy.242..349C}. Finally, suitable quotients of the flat Kasner spacetime,
\be\label{UniversalKasner}
g=-dt^{2}+t^{2}dx^{2}+dy^{2}+dz^{2},\quad (t,x,y,z)\in (0,\infty)_{t}\times \mathbb{R}^{3}_{x,y,z},
\ee
by a three-dimensional lattice in the $x,y$ and $z$ directions are instances of (\ref{iv}). Of course in this case the quotient must be in such a way that for every given $t_{0},y_{0},z_{0}$ the projection of the line $x\rightarrow (t_{0},x,y_{0},z_{0})$ is dense inside the Cauchy surface $\{t=t_{0}\}$. We do not know at the moment if every Seifert manifold, every lens space, and every $\mathbb{T}^{2}$-bundle over $\mathbb{S}^{1}$ is indeed the Cauchy horizon of a vacuum spacetime, (the complicated case seems (\ref{i})). 

Topological constraints compatible with those of Corollary \ref{cprinc} were first obtained in the interesting work of Alan Rendall \cite{1995dg.ga....10002R} where it was shown the rather strong property that Cauchy horizons can be collapsed in volume with finite curvature and diameter (for a detailed account on such constraints see \cite{MR950552}). More recently, and regarding the orbital structure of the generators, Isenberg and Moncrief \cite{moncrief2020symmetries} have shown that, for analytic spacetimes at least, there is a trichotomy as follows: either, (i) all the generators are closed, or (ii) a dense set of generators are dense in two-tori, or, (iii) there is at least one dense generator in $\Hr$. Theorem \ref{princ} and Corollary \ref{cprinc}, enclose and refine all the conclusions in these works.
\vs

The proofs of Theorem \ref{princ} and Corollary \ref{cprinc} combine four pieces of information. First, they use that the temperature of every compact non-degenerate Cauchy horizon can be normalized to a non-zero constant. This is a new result proved by the authors in \cite{reiris2020existence}, and means that there is a nowhere zero vector field $V$ on $\Hr$ tangent to the null generators such that $\nabla_{V}V=-V$. Note that, as $V$ is nowhere zero, the orbits of $V$ are the orbits of the null generators. Second, they use an important observation due to Oliver Petersen, stating that the Riemannian metric on $\Hr$ given by $\sigma=h+\omega\otimes \omega$ has $V$ as a Killing field, where here $h$ is the degenerate metric on $\Hr$ inherited from $g$, and $\omega$ is the one-form on $\Hr$ defined by $\nabla_{X}V=:\omega(X)V$, (see \cite{moncrief2020symmetries}). This follows from the well known fact that $\mathcal{L}_{V}h=0$ (holding for any vector field $V$) plus the invariance of $\omega$ under the flow of $V$, namely that $\mathcal{L}_{V}\omega=0$. We provide a simple computation of this last invariance in Proposition \ref{OPO}. The Riemannian metric $\sigma$ is fundamental to describe the orbital types of $V$, hence of the null generators, by means of standard results on isometric actions of Lie groups. So the last pieces of information have to do with that and are the following. The first is a nice observation in Riemannian geometry, recalled by Isenberg and Moncrief in \cite{isenberg1992spacetimes} (Proposition 1), and stating that the closure of the Abelian group of isometries generated by $V$ is a compact connected Abelian Lie group $G$, hence isomorphic to a torus $\mathbb{T}^{n}$, with $n\geq 1$. We review such argument for the sake of completion. The second and fourth piece of information is a result about isometric actions by Lie groups, describing the orbital structure in terms of the so called principal orbits. We take this result from \cite{alexandrino2015lie} and apply it in Theorem \ref{rm} in the next section. The main Theorem \ref{princ} and the main Corollary \ref{cprinc} then follow directly applying Theorem \ref{rm} to the isometric action of $G\sim \mathbb{T}^{n}$ on $(\Hr,\sigma)$. Theorem \ref{princ} and Corollary \ref{cprinc} are proved in the third, and last, section. 
\vs

The results in this paper have also consequences on the number of $\mathbb{T}^{n}$-symmetries that the spacetime $(M,g)$ has. Indeed, a recent important result by Petersen and R\'acz has shown that the field $V$ on $\Hr$ (shown to exist in \cite{reiris2020existence}) can be extended to a spacetime Killing field inside the globally hyperbolic region (altogether these two results answered a conjecture by Isenberg and Moncrief, see \cite{reiris2020existence}). If the spacetime falls into the class (\ref{i}), then it is easy to see that there is a $\mathbb{T}^{1}=\mathbb{S}^{1}$ spacetime symmetry. On the other hand, for spacetimes falling into the cases (\ref{ii}) and (\ref{iii}), there is instead a $\mathbb{T}^{2}$-symmetry as was shown in Corollary 1.1 in \cite{petersen2018symmetries}. We will show here that in the last case (\ref{iv}) there is a $\mathbb{T}^{3}$ symmetry and, furthermore, that the spacetime is a quotient of the Kasner spacetime as described earlier. As already mentioned, this answers a question of Isenberg and Moncrief for the so called ergodic Cauchy horizons (see Sec. D in \cite{moncrief2020symmetries}). The result can be generalised to higher dimensions without difficulty, so we state and prove it in $n+1$.

\begin{cor}\label{Kasner}
Let $\Hr$ be an ergodic, compact and non-degenerate Cauchy horizon inside a time-oriented, $n+1$-dimensional vacuum spacetime $(M,g)$. Then, $(M,g)$ is a quotient of the flat Kasner spacetime (\ref{UniversalKasner}).
\end{cor}
The proof of this corollary is given in the last section.
\vs

\noindent {\bf Acknowledgements}. The second author is greatly indebted to Oliver Petersen for letting him know about the important Riemannian metric $\sigma$ on $\Hr$ invariant under the flow of the vector field $V$ normalizing the surface gravity to a constant.

\section{Classification of the orbits of a Killing field}

Let $(M,g)$ be a $d$-dimensional smooth compact Riemannian manifold and suppose $V$ is a Killing vector field. Let $\varphi: M  \times \mathbb{R} \to M$ be the smooth flow defined by $V$, that is, the solution to the ODE, $d\varphi(p,z)/dz=V(\varphi(p,z))$, $\varphi(p,0)=p$, where $z$ is the parameter of the integral curves of $V$. For any $z\in \mathbb{R}$, let $\varphi_z: M\rightarrow M$ be the diffeomorphism given by $\varphi_{z}(p):=\varphi(p,z)$. Since $V$ is a Killing field and $\varphi_{s+z}=\varphi_{z}\circ\varphi_{s}=\varphi_{s}\circ\varphi_{z}$, then $\{\varphi_{z}:z\in \mathbb{R}\}$ is an Abelian subgroup of the group of isometries $\ism(M)$ of $(M,g)$.

In the arguments below we will use a few times the following well known facts. First, by the Myers-Steenrod theorem, the isometry group of a smooth compact Riemannian manifold is a compact Lie group \cite{myers1939group}. Second, by Cartan's theorem, a closed subgroup of a compact Lie group is a Lie group, (see Theorem 20.12 in \cite{Lee2012}). Third, by the classification theorem for Abelian Lie groups, an Abelian connected and compact Lie group is isomorphic to a torus, (see, Theorem 1.41 in \cite{alexandrino2015lie}). Finally, and fourth, the orbit of a point by the action of a compact Lie group on a complete manifold, is always an embedded submanifold \cite{alexandrino2015lie}.

We denote by $\cl(A)$ the closure of a subset $A$ of a manifold. 

\begin{prop}[From \cite{isenberg1992spacetimes}] Let $H$ be an Abelian connected subgroup of a compact Lie group $G$. Then $\cl(H)$ is an Abelian, connected, compact Lie subgroup of $G$, hence isomorphic to $\mathbb{T}^{n}$, for some $n\geq 1$.
\end{prop}

\begin{proof}
Since the closure of a connected set is connected, then $\cl(H)$ is connected. As $G$ is compact and $\cl(H)$ is closed then $\cl(H)$ is compact. Let $g = \lim_{i \to \infty} g_i$ and $h = \lim_{i \to \infty} h_i$ with $g_i\in H$, and $h_i\in H$. Then, since the group is Abelian and multiplication is continuous we have $hg=gh= \lim_{i \to \infty} g_i h_i\in  \cl(H)$. Hence $\cl(H)$ is an Abelian subgroup of $G$. Thus $\cl(H)$ is a compact, connected Abelian subgroup of $G$, hence a Lie group and thus isomorphic to $\mathbb{T}^{n}$.
\end{proof}

If $H$ is a group acting over a Riemannian manifold $(M,g)$, then the $H$-orbit of a point $p \in M$, $\{h.p:h \in H\}$, will be denoted by $\orb_H(p)$. Its isotropy group at $p$, $ \{h \in H: h.p=p\}$, will be denoted by $H_p$.

The following proposition is well known but we include the proof for the sake of completeness.

\begin{prop}\label{POLO} Suppose that $H$ is a compact connected abelian subgroup of $\mathcal{I}(M)$ acting on $(M,g)$ transitively. Then, $(M,g)$ is a flat torus\footnote{That is, the quotient of $\mathbb{R}^{n}$ by a lattice.}.
\end{prop}
\begin{proof}
First note that $H_{p}=H_{q}$ for any $p$ and $q$ in $M$. Indeed, if $h.p=p$ and $q=g.p$ then $h.q=h.(g.p)=(h.g).p=(g.h).p=g.(h.p)=g.p=q$ so $h\in H_{q}$. Thus $H_{p}\subset H_{q}$, and reversing the role of $p$ and $q$, we get $H_{p}=H_{q}$. Hence $H/H_{p}$ acts freely and transitively by isometries on $(M,g)$. Now, $H/H_{p}$ is a compact, connected and Abelian Lie group, hence isomorphic to $\mathbb{T}^{n}$ for some $n\geq 1$. We can then say that $\mathbb{T}^{n}=\mathbb{S}^{1}_{\theta_{1}}\times\ldots\times \mathbb{S}^{1}_{\theta_{n}}$ acts freely and transitively on $(M,g)$ by isometries. Denote $\Theta:=(\theta_{1},\ldots,\theta_{n})$. Fixed $\Theta$, the map $p\rightarrow \Theta.p$ is an isometry, therefore the vector fields $X_{i}$, $i=1,\ldots,n$, on $M$, given by, 
\be
X_{i}(p)=\frac{d}{d\theta_{i}} \Theta.p\, \bigg|_{\Theta=0},
\ee
are Killing fields. Slightly abusing notation, we compute,
\be\label{X}
\frac{d}{d\theta_{i}} \Theta.p=X_{i}(\Theta.p)=d_{p}\Theta (X_{i}(p)).
\ee
In particular, if $X_{i}(p)=0$ at some $p$, then $X_{i}$ is identically zero on $M$ by the last equality. In such case, the map $\theta_{i}\rightarrow \theta_{i}.p$ is constant by the first equality, i.e. $\theta_{i}.p=p$ for all $\theta_{i}$, contradicting that the action is free. Now, the result of moving $p$ by $X_{i}$ an amount $\theta_{i}$, and then by $X_{j}$ an amount $\theta_{j}$, is $\theta_{j}\theta_{i}.p$, whereas the result of moving $p$ first by $X_{j}$ and then by $X_{i}$ is $\theta_{i}\theta_{j}.p$. Since $\mathbb{T}^{n}$ is abelian we conclude that the Killing fields $X_{i}$ and $X_{j}$ commute, for all $i,j$. Using (\ref{X}) again, it is deduced that the $X_{1}(p),\ldots,X_{n}(p)$ are linearly independent at all $p$. They define thus local Euclidean coordinates systems, proving that $(M,g)$ is flat. Finally, the map $\mathbb{T}^{n}\rightarrow M$, given by $\Theta\rightarrow \Theta.p$ is bijective and non-singular, showing that $M$ is diffeomorphic to $\mathbb{T}^{n}$.
\end{proof}

\begin{prop}\label{CLAU} Let $(M,g)$ be a smooth compact Riemannian manifold. Let H be a subgroup of $\ism(M)$, and let $p$ be a point in $M$. Then, 
\be
\orb_{\cl(H)} (p)= \cl(\orb_H(p)).
\ee
Furthermore, if $H$ is connected and Abelian, then for any $p\in M$, $\orb_{\cl(H)}(p)$ is either a point or an embedded flat torus. Finally, for any $q\in \orb_{\cl(H)}(p)$, the orbit $\orb_{H}(q)$ is dense in $\orb_{\cl(H)}(p)$. 
\end{prop}
\begin{proof}

$\orb_{\cl(H)}(p)$ is closed because $\cl(H)$ is compact and group multiplication is continuous. Then, as $\orb_{H}(p)\subset \orb_{\cl(H)}(p)$ it follows that $\cl(\orb_H(p))\subset \cl(\orb_{\cl(H)}) = \orb_{\cl(H)}(p)$. To prove the other inclusion, fix $q \in \orb_{\cl(H)} (p)$. Let $\sigma \in \cl(H)$ such that $\sigma(p)=q$ and let $\sigma_i\in H$ such that $\lim_{i \to \infty} \sigma_i = \sigma$. Since the action of $\ism(M)$ on $M$ is continuous then $\sigma_i.p\rightarrow \sigma.p=q$. Therefore $q \in \cl(\orb_H(p))$.

If $H$ is connected and Abelian then $\cl(H)$ is compact, connected and Abelian and therefore isomorphic to a torus $\mathbb{T}^n$. Now, for any $p$, $\cl(H)$ acts on $\orb_{\cl(H)}(p)$ by isometries and transitively and by Proposition \ref{POLO}, $\orb_{\cl(H)}(p)$ is either a point or a flat embedded torus.

Finally, if $q\in \orb_{\cl(H)}(p)$ then there is $\sigma_i \in H$ such that $\sigma_i .p \rightarrow q$, so we have $\sigma_i^{-1}.q \rightarrow p$. Therefore $p\in \cl(\orb_{H}(q))=\orb_{\cl(H)}(q)$ and thus $\orb_{\cl(H)}(p)\subset \orb_{\cl(H)}(q)\subset \orb_{\cl(H)}(p)$.
\end{proof}

We recall now from \cite{alexandrino2015lie} the relevant notion of principal orbits and their properties. Let $H$ be a Lie subgroup of $\mathcal{I}(M)$ and $p$ a point in $M$. Let $T_{p}\orb_{H}(p)$ be the tangent space to $\orb_{H}(p)$ at $p$ and let $N_{p}$ be the perpendicular complement, so that $T_pM = T_p \orb_H(p) \oplus N_{p}$. Now, given $h\in H$, the map $p\rightarrow h.p$ is an isometry of $(M,g)$. If $h\in H_{p}$ then $h.p=p$ and so $dh.:T_{p}M\rightarrow T_{p}M$ is a linear isometry. As $T_{p} \orb_{H}(p)$ is invariant, then so is $N_{p}$. This induces an action $H_{p}\times N_{p}\rightarrow N_{p}$ called the {\it slice representation}. If this action is trivial then the orbit $\orb_{H}(p)$ is said to be {\it principal} (see Definition 3.73 and exercise 3.77 in \cite{alexandrino2015lie}). Principal orbits exist, have maximal dimension among the orbits of $H$ and the set $M_0$ defined as the union of such orbits is open and dense in $M$. Moreover, $M_0/H \subset M/H$ is a connected manifold. This is the {\it Principal Orbit Theorem} and a proof of it can be found in Theorem 3.82 of \cite{alexandrino2015lie}. An orbit that is not principal but has the same dimension as principal orbits is said exceptional. A non-exceptional and non-principal orbit is said to be singular. 

We are now ready to prove the main theorem of this section.

\begin{teo}\label{rm} Let $(M,g)$ be a smooth, 3-dimensional, compact and connected Riemannian manifold. Suppose that $V$ is a nowhere vanishing Killing vector field. Then, either,
\begin{enumerate}[\rm (I)]
\item\label{I} every orbit is closed, or,
\item\label{II} there are only two closed orbits, and every other orbit densely fills an embedded two-torus, or, 
\item\label{III} every orbit densely fills an embedded two-torus, or,
\item\label{IV} every orbit is dense in $M$.
\end{enumerate}
\end{teo}

\begin{proof} Let $H$ be connected Abelian group generated by $V$. Let $\cl(H)=:G$ that we know is isomorphic to $\mathbb{T}^{n}$, for some $n\geq 1$. Observe that as $V$ has no zeros then every $G$-orbit has dimension at least one. 

If the dimension of the principal $G$-orbits is one, then all the $G$-orbits have dimension one and are diffeomorphic to $\mathbb{S}^{1}$. Hence the $H$-orbits are closed and we are in case (\ref{I}).

If the dimension of the principal $G$-orbits is two, then $M_{0}/G$ is a connected one-manifold, therefore diffeomorphic to either $(0,1)$ or $\mathbb{S}^{1}$, and dense in $M/G$. Also, by Proposition \ref{CLAU} the principal fibers are two-tori. If $M_{0}/G$ is diffeomoprhic to $\mathbb{S}^{1}$ then every $G$-orbit is principal, therefore every $H$-orbit is dense in a two-tori, and we are in case (\ref{III}).

Let us assume that $M_{0}/G$ is diffeomorphic to $(0,1)$. We claim that in this case every non-principal orbit must be diffeomorphic to $\mathbb{S}^{1}$ (hence singular) and there must be only two of them. Suppose that $\orb_G(p)$ is a non-principal $G$-orbit of dimension two (i.e. a regular orbit). By Proposition \ref{CLAU}, $\orb_{G}(p)$ is an embedded two-torus. Let $\epsilon>0$ be small enough that the set of points at a distance $\epsilon$ from $\orb_{G}(p)$ are two embedded tori, $T_{\epsilon}$ and $T_{-\epsilon}$, at both sides of $\orb_{G}(p)$. Assume that, in addition, $\epsilon$ is chosen such that one of the tori contains a point $q$ whose $G$-orbit is principal (recall $M_{0}$ is dense). Then, such orbit must be either $T_{\epsilon}$ or $T_{-\epsilon}$ because $T_{\epsilon}\cup T_{-\epsilon}$ is preserved under the action of the isometries of $G$. Assume then that $T_{\epsilon}$ is principal. Let $\gamma$ be a length minimizing geodesic between $\orb_{G}(p)$ and $T_{\epsilon}$, starting at a point $r\in \orb_{G}(p)$ and ending at a point $s_{\epsilon}\in T_{\epsilon}$. If $\orb_{G}(p)$ is non-principal then there must exist $g\in G$ sending $\gamma'(0)$ to $-\gamma'(0)$, hence sending $s_{\epsilon}$ to a point $s_{-\epsilon}$ in $T_{-\epsilon}$. Therefore $T_{\epsilon}$ itself is not a $G$-orbit, reaching a contradiction. Thus, every non-principal orbit is diffeomorphic to $\mathbb{S}^{1}$. Let us see that there must be only two of them. Let $\orb_{G}(p)$ be an orbit diffeomorphic to $\mathbb{S}^{1}$. Let $\epsilon_{0}>0$ be small enough that for any $\epsilon\in (0,\epsilon_{0}]$ the set of points at a distance $\epsilon$ from $\orb_{G}(p)$ is an embedded two-torus $T_{\epsilon}$. Say $\epsilon_{0}$ is chosen such that $T_{\epsilon_{0}}$ is principal. Then, for any $0<\delta<\epsilon_{0}$ let $\gamma$ be a length minimizing geodesic between $T_{\epsilon_{0}}$ and $T_{\epsilon=\epsilon_{0}-\delta}$, starting at $r\in T_{\epsilon_{0}}$ and ending at $s\in T_{\epsilon}$. As the point $r$ orbits all over $T_{\epsilon_{0}}$ under the action of $G$, the point $s$ orbits all over $T_{\epsilon}$. Thus, the orbit of $s$ is two-dimensional, hence principal. It follows that near $\orb_{G}(p)$ every orbit is principal. In particular, the $G$-quotient of the $\epsilon_{0}$-tubular neighbourhood of $\orb_{G}(p)$ is naturally diffeomorphic to $[0,\epsilon_{0})$. We conclude that $M/G$ must be diffeomorphic to $[0,1]$ and that there are only two singular orbits, one over $\{0\}$ and the other over $\{1\}$. We are thus in case (\ref{II}).

If the dimension of the principal orbits is three then there is only one principal orbit equal to $M$ itself and must be a flat three-torus by Proposition \ref{CLAU}. Also by this proposition, every $H$-orbit is dense in $M$.
\end{proof}

As a consequence of the arguments shown in the proof, the topology of such spaces can be classified. For the definition of Seifert manifold, lens space, and torus bundle see \cite{MR705527}.
\begin{cor}\label{cr}
Under the hypothesis and notation of Theorem \ref{rm}, we obtain the following topological classification. 
\begin{enumerate}[\rm (I')]
\item If {\rm (\ref{I})} holds, then $M$ is a Seifert manifold.
\item If {\rm (\ref{II})} holds, then $M$ is a lens space,
\item If {\rm (\ref{III})} holds, then $M$ is a  $\mathbb{T}^2$-bundle over $\mathbb{S}^1$,
\item If {\rm (\ref{IV})} holds,  $M$ is a three-torus $\mathbb{T}^{3}$.
\end{enumerate}
\end{cor}
\begin{proof}
Cases (\ref{I}), (\ref{III}), and (\ref{IV}) are immediate. Also, if (\ref{II}) holds, then $M$ results after gluing the two solid tori $\pi^{-1}([0,1/2])$ and $ \pi^{-1}([1/2,1])$ along their boundaries and is therefore a lens space.
\end{proof}

\section{Proofs of the main results}
As earlier let $\Hr$ be a Cauchy horizon inside a smooth, vacuum, time-oriented $n+1$-dimensional spacetime $(M,g)$. Assume that $\Hr$ is compact and non-degenerate. Let $h$ denote the degenerate metric on $\Hr$ inherited from $g$.

Let us mention a few well known facts about null vector fields $Z$, tangent to $\Hr$.
First, since $\Hr$ is totally geodesic, then given any such nowhere zero $Z$ we have,
\be
g(\nabla_X Z, Y)=0,
\ee
for any $X,Y \in T\Hr$. This permits the definition of a smooth one-form $\omega_Z$ by,
\be
\nabla_X Z=: \omega_Z(X)Z,
\ee
for every $X$ in $T\Hr$. A crucial property of the forms $\omega_Z$ is that they are \textit{null closed}, meaning that,
\be
d\omega_Z(Z,\cdot) = 0.
\ee
(see for instance \cite{moncrief2020symmetries} and references therein). Second, for any such nowhere zero $Z$ it is,
\be
\mathcal{L}_{Z}h=0,
\ee
which follows after,
\be\label{eq1}
\mathcal{L}_Z h (X,Y)=  g(\omega_Z(X)Z,Y)+g(\omega_Z(Y),X)=0,
\ee
for any $X$ and $Y$ tangent to $\Hr$.

Now, it was recently shown in \cite{reiris2020existence} that one can always find a smooth, nowhere vanishing null vector field $V$ in $\Hr$ such that,
\be
\nabla_V V = -V.
\ee
This implies that $\omega_V(V)=-1$, and therefore that $\omega_V$ is nowhere vanishing. The Petersen's metric $\sigma$ (see (5) in \cite{petersen2018wave}) is the Riemannian metric over $\Hr$ defined by,
\be
\sigma(X,Y) = h(X,Y) + \omega_V(X)\omega_V(Y).
\ee
The crucial fact proved by Petersen is that $V$ is a Killing field for $\sigma$. A shorter proof of this fact than the one presented in Theorem 1.14 of \cite{petersen2018wave} can be given using Cartan's formula as shown in the next Proposition.
\begin{prop}[Petersen \cite{petersen2018wave}]\label{OPO} The vector field $V$ is a Killing vector field of $(\Hr,\sigma)$.
\end{prop}
\begin{proof}
Applying Cartan's formula, for any pair of vector fields $X,Y \in T\Hr$ we find,
\begin{align}
\nonumber \mathcal{L}_V \sigma (X,Y) &= \mathcal{L}_V h(X,Y) + \mathcal{L}_V (\omega_V \omega_V) (X,Y) \\
\nonumber &= \mathcal{L}_V (\omega_V)\omega_V(X,Y) + \omega_V  \mathcal{L}_V \omega_V (X,Y) \\
\nonumber &=  \big[\big((\iota_V \circ d) \omega_V + (d \circ \iota_V) \omega_V \big) \omega_V  + \omega_V \big((\iota_V \circ d) \omega_V + (d \circ \iota_V) \omega_V \big) \big](X,Y)\\
\nonumber &= \big[d\omega_V (V,X) +d(\omega_V(V)) (X) \big]\omega_V(Y)+\omega_V(X)\big[d\omega_V (V,Y) +d(\omega_V(V)) (Y) \big]\\
&= 0,
\end{align}
where the second equality follows from (\ref{eq1}) and the last equality follows from the properties of the one-form $\omega_{V}$ previously discussed.
\end{proof}

The proof of the Theorem \ref{princ} and of Corollary \ref{cprinc} is now straightforward from Proposition \ref{OPO} and from the Riemannian results previously obtained.
\begin{proof}[Proof of Theorem \ref{princ} and Corollary \ref{cprinc}]
Theorem \ref{princ} follows by applying Theorem \ref{rm} to $(\Hr, \sigma)$ and to the Killing vector field $V$. Corollary \ref{cprinc} now also follows immediately from \ref{cr}, without any further changes.
\end{proof}

Finally, we prove Corollary \ref{Kasner}.

\begin{proof}[Proof of Corollary \ref{Kasner}] We use Theorem 1.2 in \cite{petersen2018symmetries} to extend the vector field $V$ to a Killing vector field for $g$ in the globally hyperbolic region $H$ of $M\setminus \Hr$, (we denote the extension still with the letter $V$). Let $Z$ be the only past directed null vector field on $\Hr$ perpendicular to the kernel of $\omega$ at every point $p\in \Hr$, and such that $g(Z,V)=1$. The field $Z$ points into $H$. As $\mathcal{L}_{V}V=0$, $\mathcal{L}_{V}\omega=0$ and $\mathcal{L}_{V}g=0$ it follows that $\mathcal{L}_{V}Z=0$.

Given $p\in \Hr$ let $\gamma_{p}(s)$ be the null geodesic starting at $p$ with velocity $Z(p)$. For $\tau$ small, let $\psi_{\tau}:\Hr\rightarrow M$ be defined by $\psi_{\tau}(p)=\gamma_{p}(\tau)$. If $\tau_{0}$ is small enough then $\Sigma:=\psi_{\tau_{0}}(\Hr)$ is an embedded hypersurface in $H$ and $\psi_{\tau_{0}}:\Hr\rightarrow \Sigma$ is a diffeomorphism. As $V$ is a spacetime Killing field and $Z$ is invariant under the flow of $V$, it follows that $d\psi_{\tau_{0}}(V(p))=V(\psi_{\tau_{0}}(p))$, (recall that, as $V$ is a Killing field, it takes geodesics into geodesics). In particular $V$ is tangent to $\Sigma$. Hence, if the orbits of $V$ on $\Hr$ are dense, so are the orbits of $V$ on $\Sigma$. As $V$ is spacelike near $\Hr$ in $H$ (see, \cite{petersen2018symmetries}), it follows directly that $\Sigma$ (for $\tau_{0}$ small enough) is a spacelike hypersurface. Since $V$ is a spacetime Killing field and leaves $\Sigma$ invariant, then it is also a Killing field of the (Riemannian) metric $\overline{h}$ that $g$ induces on $\Sigma$. Thus, by Proposition \ref{POLO}, we conclude that $\mathbb{T}^{n}$ acts freely and isometrically on $(\Sigma;h)$, and that $(\Sigma,\overline{h})$ is a flat $n$-torus. Since the $\mathbb{T}^{n}$- isometries come from the spacetimes symmetries induced by $V$, it follows that the second fundamental form $K$ of $\Sigma$ is also invariant under the action of $\mathbb{T}^{n}$.

Now, any $\mathbb{T}^{n}$-invariant vacuum initial data $(\Sigma\sim \mathbb{T}^{n};\overline{h},K)$ gives rise to a Kasner spacetime but the only Kasner spacetime with a Cauchy horizon is the flat Kasner. The Corollary follows.
\end{proof}

\bibliographystyle{plain}
\bibliography{Master}

\begin{thebibliography}{10}

\bibitem{alexandrino2015lie}
Marcos~M Alexandrino and Renato~G Bettiol.
\newblock {\em Lie groups and geometric aspects of isometric actions},
  volume~8.
\newblock Springer, 2015.

\bibitem{1995AnPhy.242..349C}
P.~T. {Chrusciel} and A.~D. {Rendall}.
\newblock {Strong cosmic censorship in vacuum space-times with compact, locally
  homogeneous Cauchy surfaces.}
\newblock {\em Annals of Physics}, 242(2):349--385, September 1995.

\bibitem{Chrusciel:2003jj}
Piotr~T. Chrusciel and Kayll Lake.
\newblock {Cauchy horizons in Gowdy space-times}.
\newblock {\em Class. Quant. Grav.}, 21:S153--S170, 2004.

\bibitem{MR950552}
Kenji Fukaya.
\newblock A boundary of the set of the {R}iemannian manifolds with bounded
  curvatures and diameters.
\newblock {\em J. Differential Geom.}, 28(1):1--21, 1988.

\bibitem{isenberg1992spacetimes}
J~Isenberg and V~Moncrief.
\newblock On spacetimes containing {Killing} vector fields with non-closed
  orbits.
\newblock {\em Classical and Quantum Gravity}, 9(7):1683, 1992.

\bibitem{larsson2015smoothness}
Eric Larsson.
\newblock Smoothness of compact horizons.
\newblock In {\em Annales Henri Poincar{\'e}}, volume~16, pages 2163--2214.
  Springer, 2015.

\bibitem{Lee2012}
John~M. Lee.
\newblock {\em Introduction to Smooth Manifolds}.
\newblock Springer New York, 2012.

\bibitem{minguzzi2014completeness}
E~Minguzzi.
\newblock Completeness of {Cauchy} horizon generators.
\newblock {\em Journal of Mathematical Physics}, 55(8):082503, 2014.

\bibitem{minguzzi2015area}
E~Minguzzi.
\newblock Area theorem and smoothness of compact {Cauchy} horizons.
\newblock {\em Communications in Mathematical Physics}, 339(1):57--98, 2015.

\bibitem{1967rta1.book..160M}
C.~W. {Misner}.
\newblock {\em {Taub-Nut Space as a Counterexample to almost anything}},
  volume~8, page 160.
\newblock 1967.

\bibitem{moncrief2020symmetries}
Vincent Moncrief and James Isenberg.
\newblock Symmetries of cosmological {Cauchy} horizons with non-closed orbits.
\newblock {\em Communications in Mathematical Physics}, 374(1):145--186, 2020.

\bibitem{myers1939group}
Sumner~B Myers and Norman~Earl Steenrod.
\newblock The group of isometries of a {Riemannian} manifold.
\newblock {\em Annals of Mathematics}, pages 400--416, 1939.

\bibitem{petersen2018wave}
Oliver~Lindblad Petersen.
\newblock Wave equations with initial data on compact {Cauchy} horizons.
\newblock {\em arXiv preprint arXiv:1802.10057}, 2018.

\bibitem{petersen2018symmetries}
Oliver~Lindblad Petersen and Istv{\'a}n R{\'a}cz.
\newblock Symmetries of vacuum spacetimes with a compact {Cauchy} horizon of
  constant non-zero surface gravity.
\newblock {\em arXiv:1809.02580}, 2018.

\bibitem{reiris2020existence}
Mart{\'\i}n Reiris and Ignacio Bustamante.
\newblock On the existence of {Killing} fields in smooth spacetimes with a
  compact {Cauchy} horizon.
\newblock {\em arXiv:2006.08934}, 2020.

\bibitem{1995dg.ga....10002R}
Alan~D. {Rendall}.
\newblock {Compact null hypersurfaces and collapsing Riemannian manifolds}.
\newblock {\em eprint arXiv:dg-ga/951000}, pages dg--ga/9510002, October 1995.

\bibitem{MR705527}
Peter Scott.
\newblock The geometries of {$3$}-manifolds.
\newblock {\em Bull. London Math. Soc.}, 15(5):401--487, 1983.

\end{thebibliography}

\end{document}